\documentclass[11pt,letterpaper]{article}
\usepackage{fullpage}

\usepackage[inline]{enumitem}
\usepackage{amsfonts}
\usepackage{amsmath}
\usepackage{amsthm}
\usepackage{url,hyperref}
\usepackage{picins,boxedminipage}
\usepackage{amsfonts}
\usepackage{amssymb}
\usepackage{tikz}

\usepackage{array}

\usepackage{float}
\usepackage{graphicx}

\usepackage{makeidx}

\graphicspath{{graphics/}}

\definecolor{darkblue}{rgb}{0,0,0.38}
\definecolor{darkred}{rgb}{0.6,0,0}
\definecolor{darkgreen}{rgb}{0.1,0.35,0}
\hypersetup{colorlinks, linkcolor=darkblue, citecolor=darkgreen, urlcolor=darkblue}

\usetikzlibrary{matrix,fit,calc}
\usetikzlibrary{decorations.pathmorphing, decorations.pathreplacing}
\usetikzlibrary{shapes.misc}

\DeclareMathOperator{\conv}{conv}

\newenvironment{proofof}[1]{\noindent{\bf Proof of #1}:}{$\hfill \Box$\\}

\newtheorem{definition}{Definition}

\newtheorem{theorem}{Theorem}[section]
\newtheorem{lemma}[theorem]{Lemma}

\newtheorem{claim}{Claim}
\newtheorem{property}{Property}

\newcommand{\NP}{\textsc{NP}}
\newcommand{\RR}{\mathbb{R}}

\newcommand{\lamfam}{\mathcal{L}}

\newcommand{\lpa}{{\textrm{LPA}}}

\def\b1{{\bf 1}}

\begin{document}

\title{$k$-Trails: Recognition, Complexity, and Approximations}

\author{Mohit Singh\thanks{Microsoft Research, Redmond, USA. Email: \href{mailto:mohits@microsoft.com}{mohits@microsoft.com}\;.}\and Rico Zenklusen\thanks{ETH Zurich, Zurich, Switzerland. Email: \href{mailto:ricoz@math.ethz.ch}{ricoz@math.ethz.ch}\;.}}

\maketitle
\setcounter{footnote}{0}

\begin{abstract}
The notion of degree-constrained spanning
hierarchies, also called $k$-trails,
was recently introduced in the context
of network routing problems. They describe
graphs that are homomorphic images of
connected graphs of degree at most $k$.
First results highlight several
interesting advantages of $k$-trails compared
to previous routing approaches.
However, so far, only
little is known regarding computational
aspects of $k$-trails.

In this work we aim to fill this gap by
presenting how $k$-trails can be analyzed
using techniques from algorithmic matroid theory.
Exploiting this connection, 
we resolve several open questions
about $k$-trails.
In particular, we show that one
can recognize efficiently whether a graph is
a $k$-trail. Furthermore, we show that deciding
whether a graph contains a $k$-trail is NP-complete;
however, every graph that contains a $k$-trail
is a $(k+1)$-trail.
Moreover, further leveraging the connection
to matroids, we consider the problem of finding
a minimum weight $k$-trail contained in a graph
$G$. We show that
one can efficiently find a
$(2k-1)$-trail contained in $G$
whose weight is no more than
the cheapest $k$-trail contained in $G$,
even when allowing negative weights.

The above results settle several open questions
raised by M\'olnar, Newman, and Seb\H o.

\end{abstract}

\section{Introduction}

Motivated by applications in network routing,
the notion of
\emph{degree-constrained spanning hierarchies}
was introduced as a way
to obtain lower-degree routing structures as
what could be obtained with low-degree spanning
subgraphs~\cite{molnar_2014_new}.
These hierarchies, which, for brevity, have also
simply be called
\emph{$k$-trails}~\cite{sebo_2015_travelling,molnar_2015_travelling}, describe how a given
graph can be described as the homomorphic image
of another low-degree graph.
First results have already been reported
that show advantages of $k$-trails
compared to traditional methods in network routing
contexts~\cite{molnar_2014_new}.
However, many basic questions around $k$-trails
remained open, including whether one can efficiently
decide if a graph is a $k$-trail;
we refer the interested reader
to~\cite{sebo_2015_travelling} for
a nice overview of some open problems
around $k$-trails.
The goal of this work is to fill this gap by
answering several basic open questions
about $k$-trails, by revealing and exploiting
a connection to matroids.

Before giving a summary of our main results,
we start by formally defining
$k$-trails as well as some closely related
notions.
In particular, the notion of
homomorphic images introduced below
is the basis for defining $k$-trails.

Throughout this paper, we focus on undirected
connected graphs with possibly loops
and parallel edges, and with at least $2$
vertices to avoid trivial special cases.

\begin{definition}[homomorphic image]
A graph $G=(V,E)$ is the
\emph{homomorphic image} of a graph
$H=(W,F)$ if there is an onto function
$\phi: W \rightarrow V$ such that
for any two vertices $u,v\in V$
(with possibly $u=v$),
the number of edges in $G$ between
$u$ and $v$ is equal to the number
of edges in $H$ whose endpoints
get mapped by $\phi$ to $\{u,v\}$.
\end{definition}

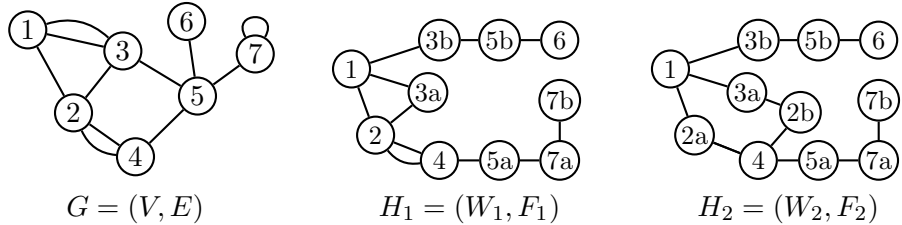
\begin{figure}[h]
\begin{center}
\begin{tikzpicture}[scale=0.8]

\tikzstyle{sn}=[thick, draw=black, fill=white, circle, minimum size=14,inner sep=1pt, font=\small]

\def\yt{-4.5}

\begin{scope}

\begin{scope}[every node/.style={sn}]
\node (1) at (1.86,-1.53) {};
\node (2) at (2.62,-2.91) {};
\node (3) at (3.44,-1.88) {};
\node (4) at (3.65,-3.64) {};
\node (5) at (4.68,-2.63) {};
\node (6) at (4.50,-1.39) {};
\node (7) at (5.64,-1.93) {};
\end{scope}

\begin{scope}
\foreach \i in {1,...,7} {
\node at (\i) {$\i$};
}
\end{scope}

\begin{scope}[thick]
\draw (1) -- (2);
\draw (1) to[bend left] (3);
\draw (1) -- (3);
\draw (2) -- (3);
\draw (2) to[bend right] (4);
\draw (2) -- (4);
\draw (3) -- (5);
\draw (4) -- (5);
\draw (5) -- (6);
\draw (5) -- (7);
\draw (7) to[out=60, in=120, looseness=4] (7);
\end{scope}

\coordinate (c1) at (4);

\end{scope}

\begin{scope}[shift={(5.7,0.3)}]

\begin{scope}[every node/.style={sn}]
\node  (1) at (1.54,-2.48) {1};
\node  (2) at (1.94,-3.58) {2};
\node  (3a) at (2.81,-2.87) {3a};
\node  (3b) at (3.00,-2.00) {3b};
\node  (4) at (3.00,-4.00) {4};
\node  (5a) at (4.00,-4.00) {5a};
\node  (5b) at (4.00,-2.00) {5b};
\node  (6) at (5.00,-2.00) {6};
\node  (7a) at (5.00,-4.00) {7a};
\node (7b) at (5.00,-3.00) {7b};
\end{scope}

\begin{scope}[thick]
\draw  (1) --  (2);
\draw  (1) --  (3a);
\draw  (1) --  (3b);
\draw  (2) --  (3a);
\draw  (2) --  (4);
\draw  (2) to[bend right] (4);
\draw  (3b) --  (5b);
\draw  (4) --  (5a);
\draw  (5a) --  (7a);
\draw  (5b) --  (6);
\draw  (7a) -- (7b);
\end{scope}

\coordinate (c2) at ($(4)!0.5!(5a)$);

\end{scope}

\begin{scope}[shift={(11,0.3)}]

\begin{scope}[every node/.style={sn}]
\node  (1) at (1.54,-2.48) {1};
\node  (2a) at (1.94,-3.58) {2a};
\node  (2b) at (3.7,-3.2) {2b};
\node  (3a) at (2.81,-2.87) {3a};
\node  (3b) at (3.00,-2.00) {3b};
\node  (4) at (3.00,-4.00) {4};
\node  (5a) at (4.00,-4.00) {5a};
\node  (5b) at (4.00,-2.00) {5b};
\node  (6) at (5.00,-2.00) {6};
\node  (7a) at (5.00,-4.00) {7a};
\node (7b) at (5.00,-3.00) {7b};
\end{scope}

\begin{scope}[thick]
\draw  (1) --  (2a);
\draw  (1) --  (3a);
\draw  (1) --  (3b);
\draw (3a) -- (2b);
\draw (2b) -- (4);
\draw  (2a) --  (4);
\draw  (2a) --  (4);
\draw  (3b) --  (5b);
\draw  (4) --  (5a);
\draw  (5a) --  (7a);
\draw  (5b) --  (6);
\draw  (7a) -- (7b);
\end{scope}

\coordinate (c3) at ($(4)!0.5!(5a)$);

\end{scope}

\path let \p1 = (c1) in node at (\x1,\yt) {$G=(V,E)$};
\path let \p1 = (c2) in node at (\x1,\yt) {$H_1=(W_1,F_1)$};
\path let \p1 = (c3) in node at (\x1,\yt) {$H_2=(W_2,F_2)$};

\end{tikzpicture}

\end{center}
\caption{The graph $G$ is the homomorphic image
of $H_1$ as well as $H_2$. The naming of the vertices
has been chosen to highlight the corresponding
homomorphisms, e.g., the nodes 3a and 3b in $H_1$
are both mapped to node 3 in $G$. $G$ is a $3$-trail,
because the homomorphic preimage $H_2$ has maximum
degree $3$.
}\label{fig:homImage}
\end{figure}

Figure~\ref{fig:homImage} shows an example
graph and 2 homomorphic preimages of it.

\smallskip

Hence, a preimage $H=(W,F)$ of $G$
corresponds to a graph obtained from
$G$ by splitting each of its vertices $v\in V$
into $|\phi^{-1}(v)|$ many copies. We therefore call
$|\phi^{-1}(v)|$ the \emph{$\phi$-multiplicity},
or simply \emph{multiplicity}, of $v$.

\begin{definition}[$k$-trail]
A graph $G=(V,E)$ is a $k$-trail if it is the homomorphic
image of a connected graph $H=(W,F)$ with
maximum degree at most $k$.
\end{definition}

The graph $G$ shown in Figure~\ref{fig:homImage} is
a $3$-trail since it is the homomorphic image
of $H_2$, which has maximum degree $3$.

It is not hard to see that $k$-trails can
equivalently be defined as preimages of
trees of degree at most $k$.
\begin{definition}[$k$-tree]
A graph is a $k$-tree if it is a spanning
tree of maximum degree at most $k$.
\end{definition}

\begin{lemma}\label{lem:splitIntoTree}
If $G$ is a $k$-trail then it is a homomorphic image of a $k$-tree.
More precisely, given a connected graph $H=(W,F)$ and onto function
$\phi:W\rightarrow V$ such that $G$ is the homomorphic image of
$H$ by $\phi$, we can construct efficiently a tree $H'=(W',F')$ and
onto function $\phi':W'\rightarrow W$ such that $H$ is the
homomorphic image of $H'$ by $\phi'$.
Thus, $G$ is the $\phi' \circ \phi$-homomorphic image of $H'$, and
for $v\in V$, the $\phi' \circ \phi$-multiplicity of $v\in V$
is at least the $\phi$-multiplicity of $v$.
\end{lemma}

\begin{proof}
Let $G$ be a homomorphic image of the connected graph $H=(W,F)$ with maximum degree $k$. We show how vertices of $H$ can be split step-by-step to arrive at $H'$. We only show a single splitting step transforming $H$ into $\bar{H}=(\bar{W},\bar{F})$. Assume that $H$ contains a cycle, say $C$; otherwise we can set $H'=H$.
Let $\{w_1,w_2\}$ be an edge in $C$. Let $\bar{W}=W\cup \{x\}$ where $x$ is a new vertex. Let $\bar{H}=(\bar{W},\bar{F})$ where $\bar{F}=F\cup\{\{x,w_1\}\}\setminus\{\{w_1,w_2\}\}$ and $\bar{\phi}:\bar{W}\rightarrow W$
where $\bar{\phi}(w)=w$ if $w\in W$ and $\bar{\phi}(x)=w_2$. Clearly $\bar{H}$ is connected with maximum degree $k$ and $H$ is a homomorphic image of $\bar{H}$ by $\bar{\phi}$. Since $\bar{H}$ has one more vertex than $H$, repeatedly applying this procedure will stop as soon as we split $G$ into a connected graph with $|E|-1$ vertices, in which case it has to be a tree $H'$.
\end{proof}

Many basic questions on $k$-trails remained
open. This includes the complexity status
of deciding whether a graph is a $k$-trail
for a given $k$, an open question raised
in~\cite{sebo_2015_travelling}.
Further interesting open questions
on $k$-trails that are motivated by
routing applications are linked to the
notion of whether a graph \emph{contains} a
$k$-trail, which is defined as follows.

\begin{definition}[containing a $k$-trail]
We say that a graph $G=(V,E)$ \emph{contains
a $k$-trail} if there is a set $U\subseteq E$
such that $G'=(V,U)$ is a $k$-trail.
\end{definition}

Notice that all $k$-trails are connected graphs,
since they are homomorphic images of connected
graphs. Hence, candidate edge sets $U\subseteq E$,
for $G=(V,E)$ to contain a $k$-trail $(V,U)$,
must be such that $(V,U)$ is connected.

Regarding the containment of $k$-trails, many
interesting questions remained
open~\cite{sebo_2015_travelling}.
This includes the complexity status
of deciding whether a graph contains a
$k$-trail, and questions related to
approximation algorithms for finding
minimum weight $k$-trails. In particular,
it was conjectured in~\cite{sebo_2015_travelling}
that for any nonnegative edge
weights $w:E\rightarrow \mathbb{Z}_{\geq 0}$
and any $k\geq 3$, there exists a polynomial
algorithm returning a $(2k-2)$-trail in
$G$ whose cost is not larger than the minimum
weight $k$-trail in $G$, if $G$ contains
a $k$-trail.

In this paper we are able to settle most
of the above-mentioned open questions, by
presenting a new viewpoint on $k$-trails in
terms of matroids.

\subsection{Our results}

One of our main results, whose derivation
will also be used to highlight a strong
link between $k$-trails and matroids,
is the fact that $k$-trails can be recognized
efficiently.

\begin{theorem}\label{thm:recognition}
Given a graph $G=(V,E)$ and $k\in \mathbb{Z}_{>0}$,
it can be checked efficiently whether $G$
is a $k$-trail, and if so, obtain a $k$-tree
$H=(W,F)$ and onto function $\phi:W\rightarrow V$
such that $G$ is the homomorphic image of $H$
by $\phi$.
\end{theorem}

Contrary to the recognition problem, the
containment problem is hard.

\begin{theorem}\label{thm:hardness}
For any $k\geq \mathbb{Z}_{\geq 2}$, the problem of
deciding whether a graph contains
a $k$-trail is \NP-complete.
\end{theorem}

Despite the different complexity status
of the containment and recognition question,
the following theorem shows that they
are closely related.

\begin{theorem}\label{thm:containToBeing}
If $G$ contains a $k$-trail then it is
a $(k+1)$-trail.
\end{theorem}

In particular, the above theorem implies that
if a graph $G$ contains a $k$-trail, we
can efficiently find a $(k+1)$-trail contained
in $G$, namely the graph itself. Using the
fact that recognition is polynomial time solvable,
we can thus find the smallest $k$ for
which a given graph $G$ is a $k$-trail, which
then implies by Theorem~\ref{thm:containToBeing}
that the smallest $k'$ for which $G$ contains
a $k'$-trail is either $k$ or $k-1$.

Finally, we obtain the following result
on the containment of weighted $k$-trails.

\begin{theorem}\label{thm:weighted}
There exists a polynomial time algorithm that,
given a graph $G=(V,E)$ with weight
function $w:E\rightarrow \mathbb{Z}$
and an integer $k\geq 2$, either shows that there
is no $k$-trail contained in $G$ or
returns a $(2k-1)$-trail contained in $G$ whose
total weight is at most the weight of any
$k$-trail contained in $G$.
\end{theorem}

Theorem~\ref{thm:weighted} almost resolves
a conjecture in~\cite{sebo_2015_travelling},
claiming that one can find a $(2k-2)$-trail
in $G$ of weight upper bounded by the weight
of any $k$-trail contained in $G$, assuming
$k\geq 3$ and that the weights $w$ are
nonnegative.
Our result only implies the existence
of a cheap $(2k-1)$-trail; however, it
holds for arbitrary weights, and not just
nonnegative ones.
Furthermore, we can show that, for arbitrary
weights, the factor
$2k-1$ is optimal when comparing to
a natural LP relaxation.

\subsubsection*{Organization of paper}

Section~\ref{sec:recognition} proves
Theorem~\ref{thm:recognition} and shows how
$k$-trails can be studied using tools from
algorithmic matroid theory.
Section~\ref{sec:containsKTrail} discusses hardness proof for Theorem~\ref{thm:hardness}
as well as the algorithm for Theorem~\ref{thm:containToBeing}. In Section~\ref{sec:weighted}, we give the
algorithm for Theorem~\ref{thm:weighted}.

\subsubsection*{Basic terminology}

The degree of a vertex $v\in V$
in a graph $G=(V,E)$ is denoted by
$\deg_G(v)$, or simply $\deg(v)$ if there is
no danger of confusion.
If the graph is clear from context,
we will also use the notation
$\deg_U(v) := |\delta(v)\cap U|$
for a set $U\subseteq E$ and $v\in V$.

\section{Recognition of $k$-trails}\label{sec:recognition}

Let $G=(V,E)$ be an undirected graph
and assume we want to show that $G$
is a $k$-trail for $k$ as small as possible.
Consider some tree $H=(W,F)$ such that
$G$ is the homomorphic image of $G$ by some
onto function $\phi: W\rightarrow V$.
Let $v\in V$ and consider all vertices
of $H$ that get mapped to $v$, i.e.,
$\phi^{-1}(v)=\{w_1,\dots, w_\ell\}$,
where $\ell=|\phi^{-1}(v)|$ is the
multiplicity of $v$.
Clearly, we have
\begin{equation*}
\deg_G(v) = \sum_{i=1}^\ell \deg_H(w_i).
\end{equation*}
Knowing that $v$ gets split into $\ell$
vertices by $\phi$, for degrees to be
low in $H$ it would be best if all
$w_i$ for $i\in [\ell]$ have about the
same degree.
It turns out that starting with any
$H$ and corresponding $\phi$, we can
balance out the degrees of vertices in
$H$ that correspond to the same vertex in $G$,
using a simple modification algorithm.
The following lemma formalizes this
statement.

\begin{lemma}\label{lem:balancing}
Given a graph $G=(V,E)$, a tree
$H=(W,F)$, and an onto function
$\phi: W \rightarrow V$ such that
$G$ is the homomorphic image of
$H$, one can determine in polynomial
time a tree $H'=(W,F')$
such that
\begin{enumerate}[nosep, label=(\roman*)]
\item $G$ is the homomorphic image
of $H'$ by $\phi$.

\item For any $v\in V$ and
$w\in W$ such that $\phi(w) = v$,
the degree of $w$ in $H'$ is either
$\left\lfloor
\frac{\deg_G(v)}{|\phi^{-1}(v)|} \right\rfloor$
or
$\left\lceil \frac{\deg_G(v)}{|\phi^{-1}(v)|}
\right\rceil$.
\end{enumerate}
\end{lemma}
\begin{proof}
We show that there exists a tree $H'=(W,F')$ such that for each vertex $v\in V$, and $w,w'\in \phi^{-1}(v)$, we have $|\deg_{H'}(w)-\deg_{H'}(w')| \leq   1$. The lemma then follows naturally. Suppose that $H=(W,F)$ does not satisfy the condition. Thus for some $v\in V$ and $w,w'\in \phi^{-1}(v)$, we have $\deg_{H}(w)\geq 2+ \deg_{H}(w')$. Let $u$ be a neighbor of $w$ that does not lie on the unique path from  $w$ to $w'$ in $H$. Since the degree of $w$ is at least three, such a neighbor always exists. Now, we construct a $\hat{H}=(W,\hat{F})$ where we let $\hat{F}=F\cup \{\{u,w'\}\}\setminus \{\{u,w\}\}$. It is straightforward to see that $\hat{H}$ is a tree and $G$ is also an homomorphic image of $\hat{H}$.  We now repeat this operation as long as we can find such a pair of vertices. Observe that  $\sum_{u\in W} \deg_{\hat{H}}(u)^2 <\sum_{u\in W} \deg_{{H}}(u)^2$ and thus a potential argument implies that the sequence is polynomially bounded. Thus the tree obtained at the end of the sequence, $H'=(W,F')$, satisfies the conditions of the lemma.
\end{proof}

We leverage the above lemma to rephrase
the problem of whether a graph is a $k$-trail
in terms of multiplicities.

\begin{definition}[feasible multiplicity vector]
Let $G=(V,E)$ be a graph.
A vector $\lambda \in \mathbb{Z}_{> 0}^V$ is a
\emph{feasible multiplicity vector} (for $G$)
if $G$ is the homomorphic image of a connected graph with
multiplicities given by $\lambda$;
more formally,
if there is a connected graph $H=(W,F)$ such that
$G$ is the homomorphic image of $H$ by some
onto function $\phi:W\rightarrow V$, and
$|\phi^{-1}(v)| = \lambda(v) \;\;\forall v\in V$.
\end{definition}

Feasible multiplicity vectors fulfill the following
down-monotonicity property.

\begin{lemma}\label{lem:multDownClosed}
Let $G$ be a graph and $\lambda \in \mathbb{Z}^V_{>0}$
be a feasible multiplicity vector. Then any vector
$\lambda' \in \mathbb{Z}^V_{>0}$ with $\lambda' \leq \lambda$
(component-wise) is also a feasible multiplicity
vector.

Furthermore, this result is constructive:
Given a connected graph $H$ and homomorphism $\phi$
such that $G$ is the $\phi$-homomorphic image of $H$
and $\lambda$ is the multiplicity vector corresponding
to $\phi$, we can efficiently construct
for any $\lambda'\leq \lambda$ a connected graph $H'$
and homomorphism $\phi'$ such that $G$ is the
$\phi'$-homomorphic image of $H'$ with corresponding
multiplicity vector $\lambda'$.
\end{lemma}

\begin{proof}
Given a feasible vector $\lambda\in \mathbb{Z}^V_{>0}$ with $\lambda(v)\geq 2$, we show that $\bar{\lambda}$ is also a feasible vector where $\bar{\lambda}(v)=\lambda(v)-1$ and $\bar{\lambda}(u)=\lambda(u)$ for each $u\in V\setminus \{v\}$. The lemma then follows from induction.

Let $H=(W,F)$ and $\phi:W\rightarrow V$ certify the feasibility of $\lambda$. Since $\lambda(v)\geq 2$, there exists $w_1,w_2\in \phi^{-1}(v)$ such that $w_1\neq w_2$. Let $H'=(W',F')$ be obtained by merging $w_1$ and $w_2$ into a new vertex $w$. Thus $W'=W\cup \{w\}\setminus \{w_1,w_2\}$. Let $\phi':W'\rightarrow V$ where $\phi'(w)=v$ and $\phi'(x)=\phi(x)$ for each $x\in W'\setminus \{w\}$. Then $G$ is a homomorphic image of $H'$ by $\phi'$ and $H'$ certifies the feasibility of $\bar{\lambda}$, thus proving the lemma.
\end{proof}

Lemma~\ref{lem:balancing} and Lemma~\ref{lem:multDownClosed}
easily imply that
the question of whether $G$ is a $k$-trail
for some given $k$ can be reduced to
the problem of deciding whether
some multiplicity vector
$\lambda\in \mathbb{Z}_{>0}^V$
is feasible.

\begin{lemma}\label{lem:recToMult}
A graph $G=(V,E)$ is a $k$-trail if and only
if the following multiplicity vector
$\lambda\in \mathbb{Z}^V_{>0}$ is feasible:
\begin{equation*}
\lambda(v) = \left\lceil \frac{\deg_G(v)}{k} \right\rceil
  \qquad \forall v\in V.
\end{equation*}
\end{lemma}
\begin{proof}
Let $G$ be a $k$-trail witnessed by $H$ and $\phi$.  Since $H$ has maximum degree $k$, we must have  $\phi^{-1}(v)\geq \left \lceil \frac{\deg_G(v)}{k}\right \rceil$ proving the only if direction.

Now, let $\lambda$ be feasible where $\lambda(v) = \left\lceil \frac{\deg_G(v)}{k} \right\rceil$  for each $ v\in V$. From Lemma~\ref{lem:balancing}, it follows that $G$ is a homomorphic image of a connected graph $H=(W,F)$
by $\phi:W\rightarrow V$ such that $\deg_H(w)\leq \left\lceil \frac{\deg_G(v)}{\lambda(v)}
\right\rceil $ for each $w\in \phi^{-1}(v)$. Since $\left\lceil \frac{\deg_G(v)}{\lambda(v)}
\right\rceil
 \leq k$, the lemma follows.
\end{proof}

To finally provide an efficient recognition algorithm
to decide whether a graph is a $k$-trail, we show
that feasible multiplicity vectors are highly structured.

Notice that a feasible multiplicity vector is at least
$1$ in each coordinate.
For simplicity, we introduce a shifted version of
feasible multiplicity vectors, called \emph{feasible
split vector}; a vector $\mu\in \mathbb{Z}_{\geq 0}^V$ is
a feasible split vector if $\mu+\b1$ is a feasible multiplicity
vector, where $\b1\in \mathbb{Z}^V$ is the all-ones vector.
Hence, a split vector tells us how many times a vertex
is split.

\begin{theorem}\label{thm:polym}
Let $G$ be an undirected graph.
The set of feasible split vectors correspond to
the integral points of a polymatroid~\footnote{
A polymatroid over a finite set $N$ is a polytope
$P\subseteq \mathbb{R}^N_{\geq 0}$ described by
$P=\{x\in \mathbb{R}^N_{\geq 0} \mid x(S) \leq f(S)
\;\forall S\subseteq N\}$,
where $f:2^N \rightarrow \mathbb{Z}_{\geq 0}$
is a submodular function. We refer the interested
reader to~\cite[Volume B]{schrijver_2003_combinatorial}
for more information on polymatroids.}, i.e.,
\begin{equation*}
P_G := \conv(\{\mu\in \mathbb{Z}^V_{\geq 0} \mid \mu
\text{ is a feasible split vector}\}),
\end{equation*}
is a polymatroid.

Furthermore, we can efficiently optimize over
$P_G$, and for any feasible split vector
$\mu\in \mathbb{Z}^V_{\geq 0}$ we can efficiently
find a connected graph $H=(W,F)$ and an onto function
$\phi:W\rightarrow V$ such that $G$ is the
homomorphic image of $H$, and
$|\phi^{-1}(v)| = \mu(v) \;\forall v\in V$.
\end{theorem}

Before proving the theorem, we start with
a few observations and show that Theorem~\ref{thm:polym}
implies Theorem~\ref{thm:recognition}.

It is well-known
that $P_G$ being a polymatroid implies that
$P_G$ is given by
\begin{equation*}
P_G = \{x\in \mathbb{R}_{\geq 0}^V \mid x(S) \leq f(S)
 \;\forall S\subseteq V\},
\end{equation*}
where $f:2^V \rightarrow \mathbb{Z}_{\geq 0}$ is the
submodular function defined by
\begin{equation*}
f(S) = \max\{x(S) \mid x\in P_G\} \qquad S\subseteq V.
\end{equation*}
Many results on polymatroids typically assume
that a polymatroid is given through a value oracle
for the function $f$.
Clearly, if we can efficiently optimize over $P_G$,
we can also evaluate efficiently the submodular
function $f$, which, as described above, corresponds
to maximizing a $\{0,1\}$-objective over $P_G$.

We are particularly interested in checking whether some
split vector $\mu\in \mathbb{Z}_{\geq 0}^V$ is feasible.
Having an efficient evaluation oracle for $f$ allows
for checking whether $\mu\in P_G$
by standard procedures:
It suffices to solve the submodular
function minimization problem
$\min\{f(S) - x(S) \mid S\subseteq V\}$; if the optimal
value is negative then $\mu\not\in P_G$, otherwise
$\mu\in P_G$.
We will later see that the link between $k$-trails
and matroids that we establish implies
an easy way to check whether $\mu\in P_G$
using a simple matroid intersection problem
(without relying on submodular function
minimization).

Combining the above results and observations,
Theorem~\ref{thm:recognition} easily follows.

\medskip

\begin{proofof}{Theorem~\ref{thm:recognition}}
Let $\lambda\in \mathbb{Z}^V_{>0}$
be defined as in Lemma~\ref{lem:recToMult}, and let $\mu=\lambda-\b1$.
Lemma~\ref{lem:recToMult}---rephrased in terms of $\mu$---states
that $G$ is a $k$-trail if and only if $\mu$ is a feasible
split vector, which can be checked efficiently by
Theorem~\ref{thm:polym}. Furthermore, if $\mu$ is feasible,
then Theorem~\ref{thm:polym} also shows that we can efficiently
obtain a graph $\bar{H}=(W,\bar{F})$ and onto function
$\phi:W \rightarrow V$ such that
\smallskip
\begin{enumerate}[nosep, label=(\roman*)]
\item $G$ is the homomorphic image of $\bar{H}$ by $\phi$, and
\item $|\phi^{-1}(v)|=\mu(v)+1 \quad \forall v\in V$.
\end{enumerate}
\smallskip
By Lemma~\ref{lem:balancing} we can balance the degrees
of $\bar{H}$ for each vertex set $\phi^{-1}(v)$ efficiently,
to obtain a graph $H=(W,F)$ with balanced degrees as
stated in Lemma~\ref{lem:balancing}. It remains to observe
that $H$ is indeed a $k$-tree. This indeed holds:
let $w\in W$ and $v = \phi(w)$; we thus obtain
\begin{align*}
\deg_H(w) &\leq \left\lceil \frac{\deg_G(v)}{|\phi^{-1}(v)|}
           \right\rceil
        = \left\lceil \frac{\deg_G(v)}{\mu(v) + 1} \right\rceil
        = \left\lceil
             \frac{\deg_G(v)}{\left\lceil \frac{\deg_G(v)}{k}
               \right\rceil}
          \right\rceil
   \leq \left\lceil
           \frac{\deg_G(v)}{\frac{\deg_G(v)}{k}}
         \right\rceil
      = k\enspace,
\end{align*}
where the first inequality follows by Lemma~\ref{lem:balancing}
and the second equality by
$\mu(v)+1 = \lambda(v)= \lceil \deg_G(v)/k \rceil$.
\end{proofof}

\subsection*{Matroidal description
of $k$-trails and proof of Theorem~\ref{thm:polym}}

We start by introducing an auxiliary graph
$G'=(V',E')$ such that spanning trees in $G'$
can be interpreted as graphs $H$ such that $G$
is a homomorphic image of $H$.
Using this connection, we then derive
that $P_G$ is a polymatroid over which we can
optimize efficiently, and show how to
construct a homomorphic preimage of $G$
corresponding to some split vector $\mu$,
as claimed by Theorem~\ref{thm:polym}.

Hence, let $G=(V,E)$ be an undirected graph.
The graph $G'=(V',E')$ contains a vertex
for each of the two endpoints of each edge in $E$.
More formally, for each $v\in V$, the Graph $G'$
contains $\deg_G(v)$ many vertices
$V'_v := \{v_e\}_{e\in \delta(v)}$; hence,
\begin{equation*}
V' = \bigcup_{v\in V} V'_v.
\end{equation*}
We note that describing
$V'_v$ by $\{v_e\}_{e\in \delta(v)}$ is a slight
abuse of notation, since for also for
each loop at $v$ we include two vertices in
$V_v'$ and not just one.

Furthermore,
\begin{align*}
E' &= \bar{E} \cup K, \quad \text{where} \\
\bar{E} &= \{\{v_e,u_e\}\mid e=\{u,v\}\in E\}, \quad \text{and}\\
K &= \bigcup_{v\in V} K_v, \quad \text{where} \\
K_v &= \{\{v_e, v_f\} \mid v_e,v_f\in V'_v, v_e\neq v_f\}
\quad \forall v\in V.
\end{align*}
See Figure~\ref{fig:auxGraph} for an example of
the above construction.

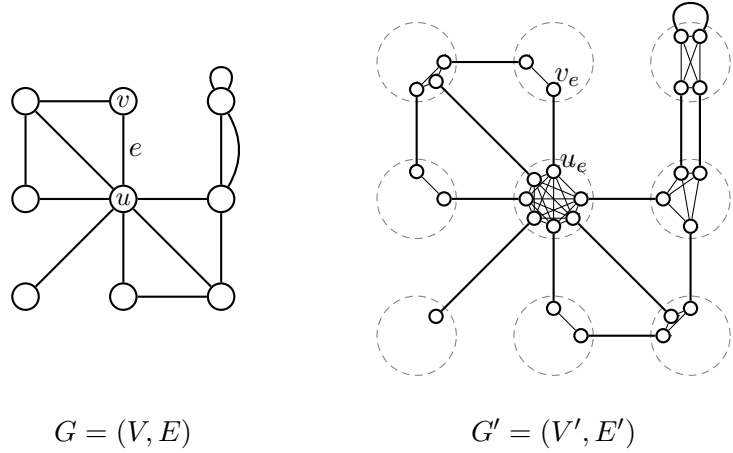
\begin{figure}[h]
\begin{center}
\begin{tikzpicture}[scale=1.3]

\tikzstyle{sn}=[thick, draw=black, fill=white, circle, minimum size=10,inner sep=1pt, font=\small]

\begin{scope}[shift={(0,0.4)}]

\begin{scope}[every node/.style={sn}]
\node (1) at (1.00,-4.00) {};
\node (2) at (2.00,-4.00) {};
\node (3) at (3.00,-4.00) {};
\node (4) at (1.00,-3.00) {};
\node (5) at (2.00,-3.00) {$u$};
\node (6) at (3.00,-3.00) {};
\node (7) at (1.00,-2.00) {};
\node (8) at (2.00,-2.00) {$v$};
\node (9) at (3.00,-2.00) {};
\end{scope}

\begin{scope}
\end{scope}

\begin{scope}[thick]
\draw (1) -- (5);
\draw (2) -- (3);
\draw (2) -- (5);
\draw (3) -- (5);
\draw (3) -- (6);
\draw (4) -- (5);
\draw (4) -- (7);
\draw (5) -- (6);
\draw (5) -- (7);
\draw (5) -- node[right=-2pt,pos=0.5] {$e$} (8);
\draw (6) -- (9);
\draw (6) to[bend right] (9);
\draw (7) -- (8);
\draw (9) to[out=60, in=120, looseness=6] (9);
\end{scope}

\end{scope}

\begin{scope}[shift={(5,-4)},scale=1.4]

\begin{scope}[every node/.style={sn,minimum size=30,draw=gray,densely dashed,thin}]
\node (b1) at (0,0) {};
\node (b2) at (1,0) {};
\node (b3) at (2,0) {};
\node (b4) at (0,1) {};
\node (b5) at (1,1) {};
\node (b6) at (2,1) {};
\node (b7) at (0,2) {};
\node (b8) at (1,2) {};
\node (b9) at (2,2) {};
\end{scope}

\begin{scope}[every node/.style={sn,minimum size=5pt}]
\def\d{0.2}
\path (b1) ++ (45:\d) node (b11) {};

\path (b2) ++ (0:\d) node (b21) {};
\path (b2) ++ (90:\d) node (b22) {};

\path (b3) ++ (90:\d) node (b31) {};
\path (b3) ++ (135:\d) node (b32) {};
\path (b3) ++ (180:\d) node (b33) {};

\path (b4) ++ (0:\d) node (b41) {};
\path (b4) ++ (90:\d) node (b42) {};

\path (b5) ++ (0:\d) node (b51) {};
\path (b5) ++ (90:\d) node (b52) {};
\path (b5) ++ (135:\d) node (b53) {};
\path (b5) ++ (180:\d) node (b54) {};
\path (b5) ++ (225:\d) node (b55) {};
\path (b5) ++ (270:\d) node (b56) {};
\path (b5) ++ (315:\d) node (b57) {};

\path (b6) ++ (70:\d) node (b61) {};
\path (b6) ++ (110:\d) node (b62) {};
\path (b6) ++ (180:\d) node (b63) {};
\path (b6) ++ (270:\d) node (b64) {};

\path (b7) ++ (270:\d) node (b71) {};
\path (b7) ++ (315:\d) node (b72) {};
\path (b7) ++ (0:\d) node (b73) {};

\path (b8) ++ (180:\d) node (b81) {};
\path (b8) ++ (270:\d) node (b82) {};

\path (b9) ++ (70:\d) node (b91) {};
\path (b9) ++ (110:\d) node (b92) {};
\path (b9) ++ (250:\d) node (b93) {};
\path (b9) ++ (290:\d) node (b94) {};
\end{scope}

\node at ($(b52)+(0.15,0.06)$) {$u_e$};
\node at ($(b82)+(0.1,0.09)$) {$v_e$};

\foreach \bn/\nn in {1/1,2/2,3/3,4/2,5/7,6/4,7/3,8/2,9/4} {
\foreach \i in {1,...,\nn} {
\foreach \j in {\i,...,\nn} {
\draw (b\bn\i) -- (b\bn\j);
}
}
}

\begin{scope}[thick]
\draw (b11) -- (b55);

\draw (b21) -- (b33);
\draw (b22) -- (b56);

\draw (b31) -- (b64);
\draw (b32) -- (b57);

\draw (b41) -- (b54);
\draw (b42) -- (b71);

\draw (b51) -- (b63);
\draw (b52) -- (b82);
\draw (b53) -- (b72);

\draw (b61) -- (b94);
\draw (b62) -- (b93);

\draw (b73) -- (b81);

\draw (b91) to[out=60, in=120, looseness=4] (b92);
\end{scope}

\end{scope}

\begin{scope}[shift={(0,-5)}]
\path let \p1 = (2) in node at (\x1,0) {$G=(V,E)$};
\path let \p1 = (b2) in node at (\x1,0) {$G'=(V',E')$};
\end{scope}

\end{tikzpicture}

\end{center}
\caption{
An example for the construction of the auxiliary
graph $G'=(V',E')$ from $G=(V,E)$.
In $G'$ the thick edges correspond to edges in $\bar{E}$
and the thin ones to edges in $K$.
The dashed gray circles correspond to the cliques
$(V'_v,K_v)$ and highlight the
link between the vertices $v\in V$
in $G$ and vertex sets $V'_v$ in $G'$ which correspond to $v$.
}\label{fig:auxGraph}
\end{figure}

For any spanning tree $T\subseteq E'$ in $G'$ that
contains $\overline{E}$, we define a graph $H_T=(W_T,F_T)$
and an onto function $\phi_T: W_T \rightarrow V$, such that
$G$ is the homomorphic image of $H_T$ by $\phi$, as follows.
Let $K_T = T\cap K = T\setminus \overline{E}$.
For each $v\in V$ consider the connected components of
$(V'_v, K_v)$. Let $q_v$ be the number of these connected
components and let
\begin{equation*}
V'_v = C_v^1 \cup C_v^2 \cup \dots \cup C_v^{q_v}
\end{equation*}
be the partition of $V'_v$ into vertex sets of the
$q_v$ connected components in $(V'_v, K_v)$.

We now define $H_T = (W_T, F_T)$ as the graph obtained
from $G'$ by contracting all $C_v^j$ for $v\in V$ and
$j\in [q_v]:=\{1,\dots, q_v\}$.
For clarity, we call the vertices in $H_T$
\emph{nodes}.
Contracting $C_v^j$ corresponds to
replacing $C_v^j$ with a single node, which we
identify with the set $C_v^j$ for simplicity, thus
leading to the following set of nodes for $H_T$:
\begin{equation*}
W_T = \{C_v^j \mid v\in V, j\in [q_v]\}.
\end{equation*}
Furthermore, two nodes $C_v^j, C_w^\ell \in W_T$ are
adjacent if and only if there is a pair of vertices,
one in $C_v^j$ and one in $C_w^\ell$, that are
connected by an edge in $T\cap \bar{E}$;
formally, this corresponds to
the existence of $e\in E$ such that
that edge in $\bar{E}$ that corresponds to
$e$ is in $T$, and $v_e \in C_v^j$
and $w_e \in C_w^\ell$.
Moreover, $\phi_T: W_T \rightarrow V$ is defined by
\begin{equation*}
\phi_T(C_v^j) = v \quad \forall\; C_v^j \in W_T.
\end{equation*}
Figure~\ref{fig:hT} shows an example construction of
$H_T$ from a spanning tree $T$ that contains $\bar{E}$.

\begin{figure}[h]

\begin{center}
\begin{tikzpicture}[scale=1.3]

\pgfdeclarelayer{background}
\pgfdeclarelayer{foreground}
\pgfsetlayers{background,main,foreground}

\tikzstyle{sn}=[thick, draw=black, fill=white, circle, minimum size=10,inner sep=1pt, font=\small]

\begin{scope}[shift={(5,-4)},scale=1.4]

\begin{scope}[every node/.style={sn,minimum size=30,draw=gray,densely dashed,thin}]
\node (b1) at (0,0) {};
\node (b2) at (1,0) {};
\node (b3) at (2,0) {};
\node (b4) at (0,1) {};
\node (b5) at (1,1) {};
\node (b6) at (2,1) {};
\node (b7) at (0,2) {};
\node (b8) at (1,2) {};
\node (b9) at (2,2) {};
\end{scope}

\begin{scope}[every node/.style={sn,minimum size=5pt}]
\def\d{0.2}
\path (b1) ++ (45:\d) node (b11) {};

\path (b2) ++ (0:\d) node (b21) {};
\path (b2) ++ (90:\d) node (b22) {};

\path (b3) ++ (90:\d) node (b31) {};
\path (b3) ++ (135:\d) node (b32) {};
\path (b3) ++ (180:\d) node (b33) {};

\path (b4) ++ (0:\d) node (b41) {};
\path (b4) ++ (90:\d) node (b42) {};

\path (b5) ++ (0:\d) node (b51) {};
\path (b5) ++ (90:\d) node (b52) {};
\path (b5) ++ (135:\d) node (b53) {};
\path (b5) ++ (180:\d) node (b54) {};
\path (b5) ++ (225:\d) node (b55) {};
\path (b5) ++ (270:\d) node (b56) {};
\path (b5) ++ (315:\d) node (b57) {};

\path (b6) ++ (70:\d) node (b61) {};
\path (b6) ++ (110:\d) node (b62) {};
\path (b6) ++ (180:\d) node (b63) {};
\path (b6) ++ (270:\d) node (b64) {};

\path (b7) ++ (270:\d) node (b71) {};
\path (b7) ++ (315:\d) node (b72) {};
\path (b7) ++ (0:\d) node (b73) {};

\path (b8) ++ (180:\d) node (b81) {};
\path (b8) ++ (270:\d) node (b82) {};

\path (b9) ++ (70:\d) node (b91) {};
\path (b9) ++ (110:\d) node (b92) {};
\path (b9) ++ (250:\d) node (b93) {};
\path (b9) ++ (290:\d) node (b94) {};
\end{scope}

\tikzstyle{t}=[black,thick]

\begin{scope}[thin]
\draw (b21) -- (b22);
\draw (b31) -- (b32);
\draw (b31) to[bend left] (b33);
\draw (b41) -- (b42);
\draw (b52) -- (b53);
\draw (b54) to[bend left] (b56);
\draw (b55) -- (b56);
\draw (b61) -- (b63);
\draw (b63) -- (b64);
\draw (b71) to[bend left] (b73);
\draw (b72) -- (b73);
\draw (b91) -- (b94);
\draw (b92) -- (b93);
\end{scope}

\begin{scope}[t]
\draw (b11) -- (b55);

\draw (b21) -- (b33);
\draw (b22) -- (b56);

\draw (b31) -- (b64);
\draw (b32) -- (b57);

\draw (b41) -- (b54);
\draw (b42) -- (b71);

\draw (b51) -- (b63);
\draw (b52) -- (b82);
\draw (b53) -- (b72);

\draw (b61) -- (b94);
\draw (b62) -- (b93);

\draw (b73) -- (b81);

\draw (b91) to[out=60, in=120, looseness=4] (b92);
\end{scope}

\end{scope}

\begin{scope}[shift={(10,-4)},scale=1.4]

\begin{scope}[every node/.style={sn,minimum size=30,draw=gray,densely dashed,thin}]
\node (b1) at (0,0) {};
\node (b2) at (1,0) {};
\node (b3) at (2,0) {};
\node (b4) at (0,1) {};
\node (b5) at (1,1) {};
\node (b6) at (2,1) {};
\node (b7) at (0,2) {};
\node (b8) at (1,2) {};
\node (b9) at (2,2) {};
\end{scope}

\begin{scope}[every node/.style={sn,minimum size=6pt,chamfered rectangle}]
\def\d{0.2}
\def\s{0.15}
\path (b1) ++ (45:\d) node (b11) {};

\path (b2) ++ (45:\s) node (b2a) {};

\path (b3) ++ (135:\s) node (b3a) {};

\path (b4) ++ (45:\s) node (b4a) {};

\path (b5) ++ (0:\s) node (b5a) {};
\path (b5) ++ (112.5:\s) node (b5b) {};
\path (b5) ++ (225:\s) node (b5c) {};
\path (b5) ++ (300:\s) node (b5d) {};

\path (b6) ++ (0:0.5*\s) node (b6a) {};
\path (b6) ++ (110:\s) node (b6b) {};

\path (b7) ++ (315:\s) node (b7a) {};

\path (b8) ++ (180:\s) node (b8a) {};
\path (b8) ++ (270:\s) node (b8b) {};

\path (b9) ++ (0:0.5*\s) node (b9a) {};
\path (b9) ++ (180:0.5*\s) node (b9b) {};

\end{scope}

\tikzstyle{t}=[black,thick]

\begin{scope}[t]
\draw (b11) -- (b5c);

\draw (b2a) -- (b3a);
\draw (b2a) -- (b5c);

\draw (b3a) -- (b6a);
\draw (b3a) -- (b5d);

\draw (b4a) -- (b5c);
\draw (b4a) -- (b7a);

\draw (b5a) -- (b6a);
\draw (b5b) -- (b8b);
\draw (b5b) -- (b7a);

\draw (b6a) -- (b9a);
\draw (b6b) -- (b9b);

\draw (b7a) -- (b8a);

\draw (b9a) to[out=60, in=120, looseness=6] (b9b);
\end{scope}

\end{scope}

\begin{scope}[shift={(0,-5)}]
\path let \p1 = (b52) in node at (\x1,0) {$(V',T)$};
\path let \p1 = (b2) in node at (\x1,0) {$H_T=(W_T,F_T)$};
\end{scope}

\end{tikzpicture}

\end{center}
\caption{
On the left-hand side, a spanning tree $T$ in
the auxiliary graph $G'=(V',E')$, that corresponds
to the graph $G$ shown in Figure~\ref{fig:auxGraph},
is highlighted. On the right-hand side, the corresponding
graph $H_T$ is shown. The homomorphism
$\phi_T:W_T\rightarrow V$ maps
all vertices of $H_T$ within the same dashed circle
to the same vertex of $G$.
}\label{fig:hT}

\end{figure}
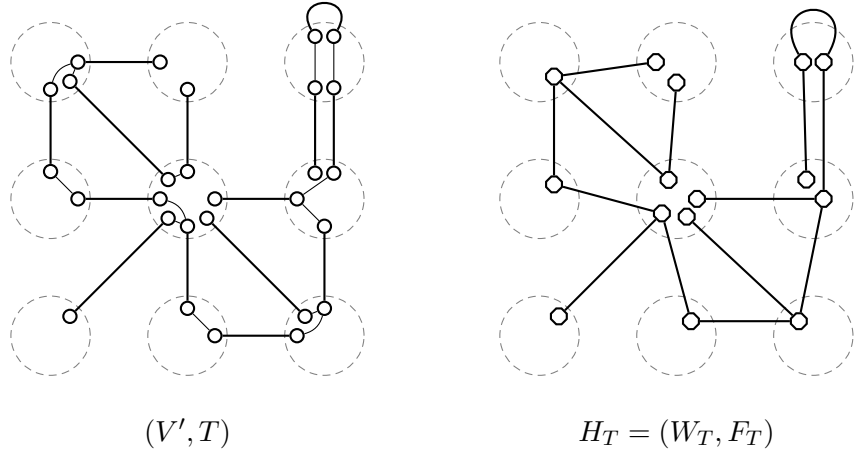

We start with some observations that follow immediately
from the above construction:
\begin{itemize}[nosep]
\item $G$ is the homomorphic image of
$H_T$ by $\phi_T$.
\item The multiplicity of $v\in V$ is $q_v$.

\item $G'$ can be constructed efficiently
from $G$.

\item Several spanning trees $T$ can lead to
the same graph $H_T$ and homomorphism $\phi_T$;
indeed, for any component $C_v^j$, the edges of
$T$ with both endpoints in $C_v^j$ form a spanning
tree over $C_v^j$, which can be replaced by any
other spanning tree without changing
$H_T$ or $\phi_T$.
\end{itemize}

\medskip

Conversely, every description of $G$ as a homomorphic
image of a tree can be obtained by the above construction:

\begin{claim}
Let $H=(W,F)$ be a tree and
$\phi:W\rightarrow V$ be an onto function
such that $G$ is the homomorphic
image of $H$ by $\phi$.
Then there exists a spanning tree $T$ in $G'$
such that $H=H_T$ and $\phi = \phi_T$.
\end{claim}
\begin{proof}
Indeed, the spanning tree $T$ can be chosen as follows.
Starting with $T=\emptyset$, we first add to $T$
all edges in $\bar{E}$.
For each $w\in W$, we do the following:
Let $v=\phi(w)$, and let $e_1, \dots, e_h\in E$ be
the $\phi$-images of the edges $\delta_H(w)$; in
particular, $e_1,\dots, e_h \in \delta_G(v)$.
We add to $T$ an arbitrary set of $h-1$ edges of
$G'$ that form a spanning tree over the vertices
$v_{e_1}, \dots, v_{e_h}$.
One can now easily observe that the constructed
tree $T$ has the desired properties.
\end{proof}

The equivalence between~\ref{item:eqTree}
and~\ref{item:eqAux} in the following statement
summarizes the above discussion. The equivalence
between~\ref{item:eqConn} and~\ref{item:eqTree}
follows from Lemma~\ref{lem:splitIntoTree}
(implying~\ref{item:eqConn} $\Rightarrow$ \ref{item:eqTree})
and Lemma~\ref{lem:multDownClosed}
(implying~\ref{item:eqTree} $\Rightarrow$ \ref{item:eqConn}).

\begin{property}\label{prop:eqAux}
Let $G=(V,E)$ be an undirected graph and
$\mu\in \mathbb{Z}_{\geq 0}^V$.
The following two statements are equivalent:
\begin{enumerate}[label=(\roman*),nosep]
\item\label{item:eqConn}
$G$ is the homomorphic image of
a connected graph $H=(W,F)$ by an onto
function $\phi:W\rightarrow V$ such that
$|\phi^{-1}(v)| = \mu(v)+1 \;\forall v\in V$.

\item\label{item:eqTree}
$G$ is the homomorphic image of
a tree $H=(V,W)$ by an onto function
$\phi: W\rightarrow V$ such that
$|\phi^{-1}(v)| \geq \mu(v)+1 \;\forall v\in V$.

\item\label{item:eqAux}
There is a spanning tree $T$ in
the auxiliary graph $G'=(V',E')$ such that
$\bar{E}\subseteq T$ and
$|T\cap K_v|
\leq |\delta(v)|-1-\mu(v) \;\forall v\in V$.
\end{enumerate}
\end{property}

Furthermore, we highlight that the equivalences
in Property~\ref{prop:eqAux} are all constructive.

Using the above connection between the auxiliary graph
$G'$ and homomorphic preimages of $G$,
Theorem~\ref{thm:polym} can now be derived as follows.
For brevity, we use the following notation.
For any spanning tree $T$ in $G'$ such
that $\bar{E}\subseteq T$, we define
$\alpha_T\in \mathbb{Z}^V_{\geq 0}$ by
$\alpha_T(v) := |T\cap K_v|
\;\forall v\in V$.
Furthermore, let $\deg\in \mathbb{Z}^V_{\geq 0}$
be the degree vector of $G$, i.e., $\deg(v)$ is the
degree of $v$ as usual.
The equivalence between point~\ref{item:eqConn}
and point~\ref{item:eqAux} of Property~\ref{prop:eqAux}
can thus be rephrased as follows.
\begin{equation}\label{eq:splitToT}
\text{$\mu\in \mathbb{Z}^V_{\geq 0}$ is a feasible
split vector}
\quad\Leftrightarrow\quad
\parbox[c]{0.40\linewidth}{
$\exists$ spanning tree $T$ in $G'$ with
$\bar{E}\subseteq T$ and $\alpha_T \leq \deg-\b1-\mu$,
}
\end{equation}
where $\b1\in \mathbb{Z}^V$ is the all-ones vector.

The equivalence highlighted by~\eqref{eq:splitToT}
directly leads to an efficient way to check whether
a given vector $\mu\in \mathbb{Z}^V_{\geq 0}$ is a
feasible split vector, and if so, obtain a homomorphic
preimage of $G$ that certifies it.
Indeed, finding a spanning tree $T$ in $G'$ that
contains $\bar{E}$ and satisfies $\alpha_T\leq \deg-\b1-\mu$
is a matroid intersection problem. More precisely, the
task is to find a spanning tree in $G'/\bar{E}$ (the
graph $G'$ after contracting $\bar{E}$)---such
spanning trees are the
bases of the graphic matroid on
$G'/\bar{E}$---whose edges
are simultaneously an independent set in the
partition matroid on the
partition $K=\cup_{v\in V} K_v$,
requiring that no more
than $\deg(v) - 1 -\mu(v)$ edges
are selected within $K_v$ for each $v\in V$.
If this matroid intersection problem has a solution,
then we get the desired spanning tree $T$ fulfilling
the conditions of point~\ref{item:eqAux} in
Property~\ref{prop:eqAux}, which is equivalent
to point~\ref{item:eqConn}, and this equivalence is
constructive, thus leading to the desired homomorphic
preimage $H$ of $G$ that corresponds to the split
vector $\mu$.

We finish by proving the claims about $P_G$ in
Theorem~\ref{thm:polym}.
For this, we start by observing that the vectors $\alpha_T$
are integral base vectors of a polymatroid.
\begin{lemma}\label{lem:alphaPolMat}
The polytope
\vspace{-0.5em}
\begin{equation*}
\bar{B}_G = \conv(\{\alpha_T \mid T \text{ is a spanning tree in $G'$
with $\bar{E}\subseteq T$}\})
\vspace{-0.5em}
\end{equation*}
is the base polytope of a polymatroid.
Furthermore, we can optimize efficiently over $\bar{B}_G$.
\end{lemma}
\begin{proof}
Let $M=(K,\mathcal{I})$ be the graphic matroid obtained
by starting with the graphic matroid on $G'$ and
contracting the edges $\bar{E}$. Notice that
the ground set of $M$ is indeed $K$, i.e.,
all edges with both endpoints in one of the set $V'_v$
for $v\in V$.
Spanning trees $T$ in $G'$ that contain $\bar{E}$
are precisely the bases of $M$.
Thus the vectors $\alpha_T$ are obtained by taking
a basis in $M$, considering the partition
$K=\cup_{v\in V} K_v$ of $K$, and counting the number
of elements that $T$ has in each part of this partition.
This way of constructing vectors out of a matroid
is often called the aggregation of a matroid
(see, e.g,~\cite{fujishige_2005_submodular}), and
is well known to lead to the integral points
of a base polytope of a
polymatroid (see \cite{fujishige_2005_submodular} and
\cite[Volume B]{schrijver_2003_combinatorial}).
This shows that $\bar{B}_G$ is indeed a polymatroid base
polytope. Furthermore, optimization over $\bar{B}_G$
can simply be done by optimizing over $M$, using
the greedy algorithm. More precisely, for a weight
function $w\in \mathbb{Z}^V$, a maximum weight point
in $\bar{B}_G$ is obtained by finding a maximum weight
spanning tree in $G'$ after contracting $\bar{E}$ and
by assigning, for $v\in V$, to each edge in $K_v$ a
weight equal to $w(v)$.
\end{proof}

Consider the polymatroid $\bar{P}_G$ that corresponds to
the base polytope $\bar{B}_G$, i.e.,
\begin{equation*}
\bar{P}_G = \{x\in \mathbb{R}^V_{\geq 0} \mid \exists
\alpha \in \bar{B}_G \text{ with } x\leq \alpha\}.
\end{equation*}
We finish the proof of Theorem~\ref{thm:polym} by showing
that $P_G$ is the (polymatroidal) dual of $\bar{P}_G$.
More precisely, McDiarmid~\cite{mcdiarmid_1975_rados}
(see also~\cite[Volume B, Section 44.6f]{schrijver_2003_combinatorial})
introduced the following notion of a dual of a polymatroid,
say $\bar{P}_G\subseteq \mathbb{R}^V$.
Consider a vector $y\in \mathcal{\mathbb{Z}^V}$ such that
$\bar{P}_G$ is contained in the box $[0,y]$; we choose
$y=\deg - \b1$.
Then the set of all points $y-\alpha$ for $\alpha\in \bar{P}_G$
correspond to the bases of a polymatroid, which is called
the \emph{dual} of $\bar{P}_G$ with respect to $y$.
By~\eqref{eq:splitToT}, the vectors $\mu\in \mathbb{Z}^V$
obtained by taking any integral point $\alpha \in \bar{B}_G$
and setting $\mu=\deg - \b1 - \alpha$ correspond precisely to
the maximal vertices of $P_G$ as defined in
Theorem~\ref{thm:polym}. Hence, $P_G$ is the dual of
$\bar{P}_G$ with respect to $y$, and thus a polymatroid.
Moreover, analogous to matroid duality, we can efficiently
optimize over $P_G$ because we can efficiently optimize
over $\bar{P}_G$
(see~\cite[Volume B]{schrijver_2003_combinatorial} for details).

\section{Containment of $k$-trails}\label{sec:containsKTrail}

We first prove the hardness result claimed in Theorem~\ref{thm:hardness}.

\medskip
\begin{proofof}{Theorem~\ref{thm:hardness}}
We give a reduction to the Hamiltonian path problem in cubic graphs which is known to be NP-complete. First we prove the theorem for $k=2$. We claim a cubic graph $G$ contains a $2$-trail if and only if $G$ contains a Hamiltonian path. Indeed if $G$ contains a Hamiltonian path $G'=(V,E')$, then the Hamiltonian path $G'$ certifies that $G$ contains a $2$-trail with $\phi$ given by the identity map between the vertices of $G'$ and $G$.

Now, suppose that $G$ contains a $2$-trail and let $G'=(V,U)$ be a $2$-trail contained in $G$ with \emph{minimum} number of edges. Let $H=(W,F)$ be a $2$-tree, i.e., a Hamiltonian path, and $\phi:W\rightarrow V$ certify that $G'$ is a $2$-trail. Let $w$ be one of the two leaves in $H$ and $v=\phi(w)$ and $(w,w')\in F$. If $|\phi^{-1}(v)|>1$, then we can consider $\hat{H}=H\setminus \{w\}$ and restrict $\phi$ to $W\setminus \{w\}$. It is straightforward to see that $\hat{G}=(V,U\setminus \{\phi(w),\phi(w')\}$ is a $2$-trail with possible preimage $\hat{H}$, thus contradicting the minimality of $G'$. Thus for each $v\in V$, if $|\phi^{-1}(v)|\geq 2$ then $\phi^{-1}(v)$ does not contain a leaf of $H$. But the degree of $v$ is three and if $|\phi^{-1}(v)|\geq 2$ then at least one of the vertices in $\phi^{-1}(v)$ must be a leaf. This implies that $\phi$ is one to one and $G'$ is a Hamiltonian path as well.

Consider any $k\geq 3$. We now reduce it to the $k=2$ case. Again consider any cubic graph $\tilde{G}=(\tilde{V},\tilde{E})$. For each vertex $v\in \tilde{V}$, we introduce a new set of vertices $v_{i}$, $1\leq i\leq k-2$ and connect them to $v$ via the edge $(v,v_i)$. We let the new graph be $G=(V,E)$ where $V=\tilde{V}\cup X$, with $X$ denoting the new vertices introduced. We now claim that $G$ has a $k$-trail if and only if $\tilde{G}$ contains a $2$-trail. Let $H=(W,F)$ denote the $k$-tree and $\phi:W\rightarrow V$ be such that $\bar{G}=(V,U)$ be the homomorphic image of $H$ by $\phi$ and $U\subseteq E$. Since every vertex $v\in X$ is a leaf in $G$, it must also be a leaf in $\bar{G}$. Thus $|\phi^{-1}(v)|=1$ and if $\{w\}=\phi^{-1}(v)$, then $w$ is also a leaf in $H$. Thus deleting $\phi(X)$ from $H$ keeps it connected. Let ${H'}=(W',F')$ be the {connected} graph obtained by deleting $\phi(X)$ from $H$. Restricting $\phi$ to $W'$, we obtain a map ${\phi'}: W'\rightarrow \tilde{V}$. Let ${G'}=(\tilde{V}, {U'})$ denote the image of ${H'}$ by ${\phi'}$ where ${U'}\subseteq \tilde{E}$. We claim that ${H'}$ is a Hamiltonian path. Suppose for sake of contradiction that there exists a vertex $w\in W'$ of degree at least three in $H'$. Let $v=\phi'(w)\in \tilde{V}$. Since $v$ has degree three in $\tilde{G}$, it must have degree exactly three in $G'$ and moreover, $|\phi^{-1}(v)|=1$. But then in $H$, the sole vertex in $\phi^{-1}(v_i)$ must be be connected to $w$ as well since $w$ is the only vertex in $\phi^{-1}(v)$. Since there are $k-2$ such vertices $v_i$, we obtain that the degree of $w$ in $H$ is at least $3+k-2= k+1$ which is a contradiction. Thus we obtain that $H$ is a Hamiltonian path and therefore, $\tilde{G}$ contains a $2$-trail, completing the proof.
\end{proofof}

\smallskip

To complement the hardness result, we prove Theorem~\ref{thm:containToBeing} which implies that we can approximate the minimum $k$ such that $G$ contains a $k$-trail by an additive one.

\medskip

\begin{proofof}{Theorem~\ref{thm:containToBeing}}
Let $G'=(V,U)$ denote a $k$-trail contained in $G$ with maximum number of edges. Let $H=(W,F)$ denote a $k$-tree and let $\phi:W\rightarrow V$ be an onto function
such that $G'$ is the $\phi$-homomorphic image of $H$.
\begin{claim}\label{claim:acyclic}
 $G\setminus U$ is acyclic.
 \end{claim}
 \begin{proof}[Proof of Claim~\ref{claim:acyclic}]
  Suppose for sake of contradiction, there is a cycle $C\in G\setminus U$. We will augment $G'$ to a $k$-trail which also includes all edges of $C$. Let $C=(v_1,\ldots, v_r)$ where $\{v_i,v_{i+1}\}\in E$ for each $1\leq i\leq r$ where we let $v_{r+1}=v_1$. Let $w_i\in \phi^{-1}(v_i)$ for each $1\leq i\leq r$. Introduce a new vertex $u_i$ for each $1\leq i\leq r$, and let $W'=W\cup \{u_1,\ldots, u_r\}$ and $F'=F\cup \{\{w_i,u_{i+1}\}: 1\leq i\leq r\}$. Moreover, we extend $\phi$ to $\phi'$ by letting $\phi'(u_i)=v_i$ for each $i$. It is easy to see that  $G''=(V,U\cup C)$ is an homomorphic image of $H'=(W',F')$ by $\phi':W'\rightarrow V$. While $H'$ is not necessarily a $k$-tree, Lemma~\ref{lem:balancing} implies that there is another $\tilde{H}=(W',\tilde{F})$ such that the degree of any vertex $w\in \phi^{-1}(v)$ for $v\in C$ is at most
  \begin{align*}
     \left \lceil \frac{\deg_{G'}(v)+2}{|\phi'^{-1}(v)|}\right\rceil \leq \left \lceil \frac{k |\phi^{-1}(v)|+2}{|\phi^{-1}(v)|+1}\right \rceil \leq k.
  \end{align*}
  For any other vertex $w$, the degree does not change and therefore, remains at most $k$. This proves
the claim.
\end{proof}

Now we will augment $U$ to $\hat{U}$ inductively one edge at a time while maintaining two hypotheses. Firstly, $\hat{G}=(V,\hat{U})$ will be a $(k+1)$-trail. Moreover, only vertices in $\hat{H}=(\hat{W},\hat{F})$---the preimage of $\hat{G}$---of degree $k+1$ will be isolated nodes in $G\setminus \hat{U}$. Clearly, the conditions are satisfied initially when $\hat{U}=U$ and $\hat{G}=G'$ since there are no vertices of degree $k+1$ in $\hat{H}=H$. Suppose it is true for some $\hat{U}\supseteq U$ such that $E\setminus \hat{U}\neq \emptyset$. Since $G\setminus \hat{U}$ is a forest, there exists a leaf vertex $u$ with the only edge incident being $\{u,v\}$. We include $\{u,v\}$ in $\hat{U}$, introduce another a new vertex $v'$ in $\hat{W}$ such that $\phi(v')=v$. Moreover, we include the edge $\{u',v'\}$ in $\hat{F}$ where $u'\in \phi^{-1}(u)$. By induction hypothesis, $u'$ must have had degree $k$ initially and after the introduction of the edge $\{u',v'\}$, its degree increases to $k+1$. Since $u$ is now isolated in $G\setminus \hat{U}$ the induction hypothesis continues to hold.

Thus, after inclusion of all edges we obtain a $(k+1)$-tree such that $G$ is the homomorphic image of the tree, proving the theorem.
\end{proofof}

\subsection{Containment of minimum weight $k$-trails}\label{sec:weighted}

Now we consider the problem of finding the minimum weight $k$-trail contained in $G=(V,E)$ and prove Theorem~\ref{thm:weighted}. Our goal is to use the auxiliary graph $G'=(V',E')$ described in the proof of Theorem~\ref{thm:recognition} for the recognition algorithm. Recall that edges in $E$ are in one-to-one correspondence with $\bar{E}\subseteq E'$.  We extend the weight function $w:E\rightarrow \mathbb{Z}$ to all edges in $E'$, where $e\in \bar{E}$ gets the same weight as the corresponding edge in $E$. The rest of the edges in $E'\setminus \bar{E}$ are assigned weight $0$.  Recall, $V_v'$ denotes the set of vertices introduced for vertex $v$ and $K_v$ denote the complete graph on $V_v'$. Identical to Property~\ref{prop:eqAux}, we state the following property. With a slight abuse of notation, for a subgraph $\hat{G}=(V,\hat{E})$ of $G$, we will also denote as $\hat{E}$ the set of edges in $\bar{E}$ that correspond to $\hat{E}$.
\begin{property}\label{prop:eqAux2}
Let $\hat{G}=(V,\hat{E})$ denote a subgraph of $G$. Let $\mu\in \mathbb{Z}_{\geq 0}^V$.
The following two statements are equivalent:
\begin{enumerate}[label=(\roman*),nosep]
\item\label{item:eqConn2}
$\hat{G}$ is the homomorphic image of
a connected graph $H=(W,F)$ of maximum degree $k$ by an onto
function $\phi:W\rightarrow V$ with
$|\phi^{-1}(v)| = \mu(v)+1 \;\forall v\in V$.

\item\label{item:eqAux2}
There is a spanning tree $T$ in
the auxiliary graph $G'=(V',E')$ such that
$\hat{E}\subseteq T$,
$|T\cap K_v|
\leq |\delta_{E}(v)|-1-\mu(v) \;\forall v\in V$, and finally, $\frac{\delta_{E'}(v)}{ |\delta_{E}(v)|- |T\cap K_v|}\leq k$.
\end{enumerate}
\end{property}

   We give a general linear programming relaxation for the problem. For any set $S\subseteq V'$ and set of edges $F\subseteq E'$, we use the notation $F(S)=\{\{u,v\}\in F: u,v\in S\}$. We introduce a variable $x_e$ for each edge $e\in E'$. The first set of constraints enforce that $x$ is in the convex hull of spanning trees of $G'$.  We place \emph{degree constraints} on $\bar{E}$-edges incident with $V_v'$ for a well chosen subset of vertices $v\in Q$. We will initialize $Q=V$ but remove these constraints successively in later iterations. We also write the linear program with the edge set $\hat{E}\subseteq E'$. Again we initialize $F=E'$. Property~\ref{prop:eqAux2} implies that the following linear program is a relaxation.
\begin{equation}
\begin{tabular}{lll@{\quad}l}
&{$\displaystyle\min \sum_{e\in E^*} w_e x_e$} &&  \\[-0.8em]
&&&  \\
&$x(E^*)$ &$= |V'|-1$ & \\
&$x(E^*(S))$ &$\leq |S| -1$ &$\forall S\subseteq V', S\neq\emptyset$\\
&$x(\delta_{E^*}(V_v)) + k x(E^*(V_v'))$ & $\leq k \cdot \deg_{E}(v)$ & $v\in Q$\\
&$x_e$&$ \geq 0$ &$\forall e\in E$\\
\end{tabular}\tag{\lpa}
\end{equation}

We now give an algorithm based on the
iterative relaxation paradigm.
Iterative relaxation and related techniques have previously
been applied to degree-constrained spanning tree
problems~\cite{singh_2007_approximating,bansal_2009_additive,bansal_2012_generalizations,zenklusen_2012_matroidal},
and we refer the reader to~\cite{lrs11}
for further details and examples related to this technique.

\begin{enumerate}[nosep]
\item Initialize $Q\gets V$, ${E^*}\gets E$.
\item While $Q\neq \emptyset$
\begin{enumerate}
\item Let ${x}$ denote the optimal extreme point solution to $\lpa$.
\item If there exists an edge $e\in {E^*}$ such that ${x}_e=0$, then ${E^*}\gets {E^*}\setminus \{e\}$.
\item If there exists a vertex $v\in V$ such that
one of the following is satisfied
\begin{align*}
\deg_{E^*}(v) + (2k-1)\cdot |E^*(V_v)| &\leq (2k-1)\cdot \deg_{E}(v)\enspace,
\textrm{ or}\\
\deg_{E^*}(v)&  \leq 2k-1\enspace,
\end{align*}
then $Q\gets Q\setminus \{v\}$.
\end{enumerate}
\item Return the optimal extreme point solution ${x^*}$ to $\lpa$.
\end{enumerate}

\smallskip

Observe that when $Q=\emptyset$, then $\lpa$ optimizes over the convex hull of spanning trees of $G'$, and therefore, $x^*$ is an indicator of a spanning tree $T$ of $G'$. Moreover, the weight of $x^*$ is at most the weight of the initial linear programming solution since we either delete edges with zero fractional value or remove constraints. For any vertex $v\in V$, let $E^*$ denote the fractional solution when $v$ was removed from $Q$.  If $\deg_{E^*}(v) \leq 2k-1$ then it follows that $\deg_{T}(v) \leq 2k-1$ since $T\subseteq E^*$. Moreover, because $|T(V_v)|\leq |V_v|-1=\deg_{E}(v)-1$, following by $T$ being a spanning tree, we have that $\deg_{T}(v) + (2k-1)\cdot |T(V_v)| \leq (2k-1)\cdot \deg_{E}(v)$, showing that the constraint for $v$ will be fulfilled for $2k-1$ even after dropping the constraint in $\lpa$ that corresponds to $v$. In the other case, we also have that $\deg_{T}(v) + (2k-1)\cdot |T(V_v)| \leq (2k-1)\cdot \deg_{E}(v)$ since $T\subseteq E^*$. Thus from Property~\eqref{prop:eqAux}, it follows that $T$ corresponds to a $(2k-1)$-trail of $G$ of weight at most the optimal fractional value proving the theorem. It remains to prove that the algorithm is able to make progress in each iteration, which is shown in the following lemma.

\begin{lemma}\label{lem:counting}
Let $x$ denote the extreme point solution to $\lpa$ with edges $E^*$ and degree constraints for $Q\subseteq V$. Then either there exists an $e\in E^*$  such that $x_e=0$ or a vertex $v\in Q$ such that one of the following is satisfied
\begin{align*}
\deg_{E^*}(v) + (2k-1)\cdot |E^*(V_v)|
  &\leq (2k-1)\cdot \deg_{E}(v)\enspace, \textrm{ or}\\
\deg_{E^*}(v)  &\leq 2k-1\enspace.
\end{align*}
\end{lemma}

The proof of Theorem~\ref{thm:weighted} now follows
from Lemma~\ref{lem:counting}. Later in Section~\ref{sec:integrality-gap},
we show  nearly matching integrality gap examples.
It remains to proof Lemma~\ref{lem:counting}.

\medskip

\begin{proofof}{Lemma~\ref{lem:counting}}
Suppose for sake of contradiction that the conditions in the lemma are not satisfied by some extreme point solution $x$. We consider the set of tight linear constraints defining $x$. Let $\tau=\{S\subseteq V': x(E^*(S))=|S|-1\} $ be the tight constraints from the spanning tree constraints and let $Q_t\subseteq Q$ denote the tight degree constraints. Standard uncrossing arguments imply the following claim (see Lau, Ravi and Singh~\cite{lrs11} Chapter 4 for details).

\begin{claim}
There exists a laminar family $\mathcal{L}\subseteq \tau$ and $Q'\subseteq Q_t$ such that the following hold.
\begin{enumerate}
\item $|\mathcal{L}|+|Q'|=|E^*|$.
\item The vectors $\{\chi(E^*(S)): S\in \lamfam \}\cup \{\chi(\delta_{E^*}(V_v)+k\chi(E^*(V_v))):v\in Q'\}$ are linearly independent.
\item $\{\chi(E^*(S)): S\in \lamfam \}$ span all the tight constraints corresponding to sets in $\mathcal{L}$ and $\{\chi(E^*(S)): S\in \lamfam \}\cup \{\chi(\delta_{E^*}(V_v)+k\chi(E^*(V_v))):v\in Q'\}$ span all the tight constraints corresponding to sets in $\mathcal{L}$ and vertices in $Q_t$.
\end{enumerate}
\end{claim}

We now show a contradiction by a counting argument. Observe that we have $|\lamfam|\leq |V'|-1$ since $\mathcal{L}$ is a laminar family over a ground set of size $|V'|$ and includes sets of size at least two. Let $z\in \RR_{\geq 0}^{E^*}$ be the slack vector defined by $z_e=1-x_e$ for any edge $e\in E^*$.
Hence, using the fact that $x(E^*)=|V'|-1$,
the total slack is bounded by
\begin{equation*}
z(E^*) = |E^*|-x(E^*) \leq |\mathcal{L}| + |Q'| - |V'| + 1 \leq |Q'|\enspace.
\end{equation*}
Since we assume the conditions of the lemma are not satisfied and constraints for vertices in $Q'$ are tight, the following inequality and equations follow
for any $v\in W'$,  where we use that $z_e+x_e=1$ for each edge $e$:

\begin{align*}
\frac{k}{2k-1} \left(\deg_{E^*}(v) + (2k-1)\cdot|E^*(V_v)| \right) &\geq \frac{k}{2k-1} \left((2k-1)\cdot \deg_{E}(v)+1\right)\enspace, \text{and} \\
-\left(x(\delta_{E^*}(v)) + k\cdot x(E^*(V_v))\right)& =- k\cdot \deg_{E}(v)\enspace, \text{and}\\
\frac{k-1}{2k-1}(x(\delta_{E^*}(v))+z(\delta_E{^*}(v)))&=
\frac{k-1}{2k-1}\deg_{E^*}(v)\enspace.
\end{align*}
Adding up the left-hand sides and right-hand sides of the
three relations above,
we obtain that for each $v\in Q'$, we have

\begin{align*}
z(\delta_{E^*}(v)) + k\cdot z(E^*(V_v)) &\geq \frac{k-1}{2k-1}|\delta_{E^*}(v))|+  \frac{k}{2k-1}\enspace.
\end{align*}
Since we assume that none of the two
conditions of Lemma~\ref{lem:counting}
holds, we have $\deg_{E^*}(v) \geq 2k$, and thus
\begin{align*}
z(\delta_{E^*}(v)) + k\cdot z(E^*(V_v)) &\geq \frac{k-1}{2k-1}\cdot 2k+  \frac{k}{2k-1}=k\enspace.
\end{align*}
Summing over vertices $v\in Q'$, we get
\begin{align*}
\sum_{v\in Q'}\left(z(\delta_{E^*}(v)) + k\cdot z(E^*(V_v)) \right) &\geq k|Q'|\enspace.
\end{align*}
Since $E^*(V_v)$ are disjoint for different vertices $v$, every edge is counted at most $k$ times in the LHS and therefore it is at most $kz(E^*)$. Moreover, we have $$z(E^*)=|E^*|-x(E^*)=|E^*|-|V'|+1\geq |E^*|-|\lamfam|=|Q'|.$$ We obtain that

\begin{align*}
k|Q'|\geq kz(E^*) &\geq k|Q'| \\
\end{align*}

Thus we must have equality everywhere and, unless $k=2$, we must have $z(\delta_{E^*}(v))=0$ for each $v\in Q'$ since these edges were counted at most twice. First assume $k\geq 3$ and thus we have $z(E^*(V_v))=1$ for each $v\in Q'$. Since $E^*(V_v)$ are disjoint, this implies that $z(E^*)\geq z(\left(\cup_{v\in Q'} E^*(V_v)\right))= |Q'|$ which is contradiction unless $z_e=0$ for each $e\in E^*\setminus \left(\cup_{v\in Q'} E^*(V_v)\right)$. Now, we show a linear dependence. Since $x_e=1$ for each $e\in E^*\setminus \left(\cup_{v\in Q'} E^*(V_v)\right)$, we have that $\chi(\{e\})$ is in the span of tight constraints in $\lamfam$. Subtracting these vectors from $\chi(E^*)$ we obtain that $\chi\left(\cup_{v\in Q'} E^*(V_v)\right)$ is in the span of tight constraints in $\tau$ which includes the vector $\chi(E^*)$. But this vector is also in the span of  $\{\chi(\delta_{E^*}(V_v)+k\chi(E^*(V_v))):v\in Q'\}$ and $\{\chi(\{e\}):e\in \cup_{v\in Q'} E^*(V_v)\}$ which is a contradiction.

In the case when $k=2$, we let $R=E^*\setminus \left(\cup_{v\in Q'} E^*(V_v)\right)\setminus \left(\cup_{v\in Q'} \delta_{E^*}(V_v)\right)$. We must have that $x_e=1$ for each $e\in R$. But summing $\{\chi(\delta_{E^*}(V_v)+2\chi(E^*(V_v))):v\in Q'\}$ we obtain $2\chi(E^*\setminus R)=2\chi(E^*)-2\chi(R)$. Since $\chi(\{e\})$ for each $e\in R$, and $\chi(E^*)$ is in the span of tight constraints in $\tau$ and therefore, $\lamfam$, we obtain a contradiction. This completes the proof of the lemma and Theorem~\ref{thm:weighted}.
\end{proofof}

\subsection{Integrality Gap Example}\label{sec:integrality-gap}

In this section, we give an integrality gap example for the linear program $\lpa$. First observe that for $k=2$, Theorem~\ref{thm:weighted} returns a 3-trail and therefore is optimal. Now consider any integer $k\geq 3$.  We show that there is an instance such that $G$ contains no $(2k-3)$-trail but the linear program is feasible for $k$. In this case, the cost function is non-negative. We also show that the if the cost function is allowed to take negative values, then any trail whose cost is at most the cost of the linear programming solution must be a $(2k-2)$-trail.

Consider the cycle $v_1,\ldots, v_n$ where we have two parallel edges between $v_i$ and $v_{i+1}$ for each $2\leq i\leq n-1$. Thus, $v_1$ has degree two, $v_2$ and $v_n$ have degree three and every other vertex has degree four. For each $1\leq i\leq n$, we add $2k-1-\deg(v)$ additional distinct vertices, say $R_i=\{w_i^1,\ldots, w_i^{b_i}\}$, and there is an edge $(w,v_i)$ for each $w\in R_i$. Thus, the degree of each vertex is now exactly $2k-1$.
Figure~\ref{fig:intGap} summarizes the construction.

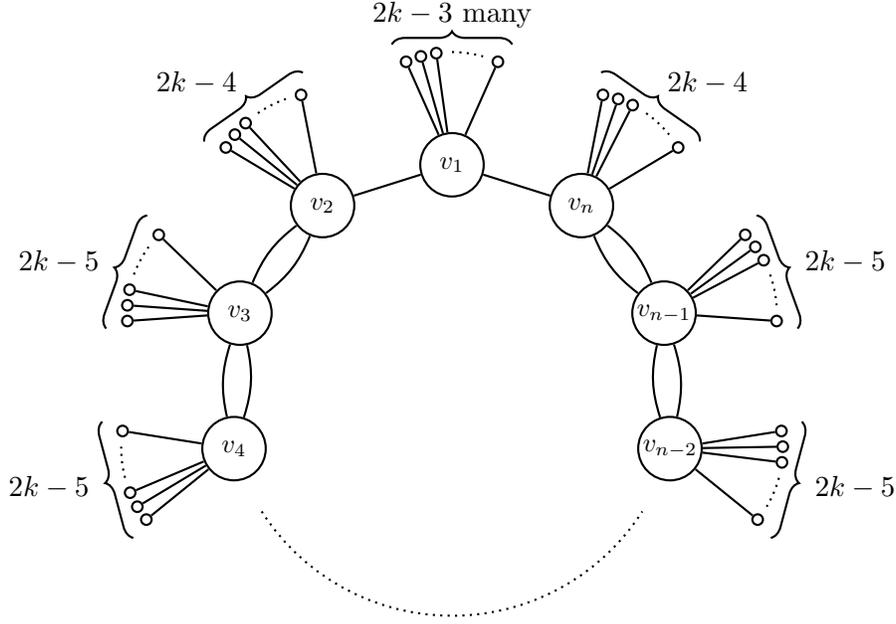
\begin{figure}[h]
\begin{center}
\begin{tikzpicture}[scale=1]

\tikzstyle{sn}=[thick, draw=black, fill=white, circle, minimum size=24,inner sep=1pt, font=\small]
\tikzstyle{pn}=[thick, draw=black, fill=white, circle, minimum size=4,inner sep=1pt, font=\small]

\def\ang{35}
\def\r{3}
\def\fanr{1.5}

\newcommand\fan[2]{
\foreach \j in {-3,-2,-1,3} {
\node[pn] (#1\j) at ($(#1) + (#2-8*\j:\fanr)$) {};
\draw[thick] (#1) -- (#1\j);
}
\draw[dotted,thick] (#1) ++ (#2:1.5) arc (#2:#2-8*2:\fanr);
}

\newcommand\fanbr[1]{
\begin{scope}[rotate around={#1*\ang:(0,0)}]
\draw[decorate, decoration={brace,amplitude=5},thick] (90:\r+\fanr+0.1) ++ (-0.8,0) -- ++(1.6,0);
\coordinate (b#1) at (90: \r+\fanr+0.5);
\end{scope}
}

\begin{scope}

\begin{scope}[every node/.style={sn}]
\foreach \i/\t in {-3/4,-2/3,-1/2,0/1,1/n,2/n-1,3/n-2} {
\node (\i) at (90-\ang*\i : \r) {$v_{\t}$};
\fan{\i}{90-\ang*\i}
}
\end{scope}

\def\bv{15}

\begin{scope}[thick]
\draw (0) -- (-1);
\draw (0) -- (1);
\draw (1) to[bend right=\bv] (2);
\draw (1) to[bend left=\bv] (2);
\draw (2) to[bend right=\bv] (3);
\draw (2) to[bend left=\bv] (3);
\draw (-2) to[bend right=\bv] (-1);
\draw (-2) to[bend left=\bv] (-1);
\draw (-3) to[bend right=\bv] (-2);
\draw (-3) to[bend left=\bv] (-2);

\draw[dotted] (90+3.5*\ang:1*\r) arc (90+3.5*\ang : 90-3.5*\ang+360 : 1*\r);

\end{scope}

\begin{scope}
\foreach \i in {-3,...,3} {
\fanbr{\i}
}
\end{scope}

\begin{scope}
\node at (b0) {$2k-3$ many};
\node[left=-4] at (b1) {$2k-4$};
\node[right=-4] at (b-1) {$2k-4$};
\node[left=-4] at (b2) {$2k-5$};
\node[right=-4] at (b-2) {$2k-5$};
\node[left=-4] at (b3) {$2k-5$};
\node[right=-4] at (b-3) {$2k-5$};
\end{scope}

\end{scope}

\end{tikzpicture}

\end{center}
\caption{The construction of our integrality
gap example. Each vertex has degree $2k-1$.
Apart from vertices $v_1$, $v_2$, and $v_n$,
all other vertices have the same number of
attached pending edges, namely $2k-5$.}\label{fig:intGap}
\end{figure}

We consider two cost functions, first when all costs are $-1$. We claim there is a feasible fractional solution with degree $k$. We select all the edges of the graph $G$. It is easy to check that $\mu$ is a feasible split vector if it has $n-1$ ones and $0$ once.
Thus, the degree of each vertex is $2k-1$ if it split and $k$ or $k-1$ if it is not split. In the uniform convex combination of these integral solutions,
each vertex of the will be split fractional $1-\frac{1}{n}$ units.
In other words, the corresponding solution to $\lpa$ satisfies
$\deg(v)-1 - x(E(V'_v)) = 1-\frac{1}{n}$ for each ring-vertex $v$.
One can easily verify that this implies feasibility for $\lpa$ with
parameter $k$ when $n\geq k$.
Since the weights are $-1$ for each edge, we must include each edge in the integral solution to obtain a weight no higher than the fractional solution. Since $\mu=(1,\ldots,1)$ is not a feasible vector, we will not be able to split each vertex. But any vertex that is not split must have degree $2k-1$, giving the desired lower bound.

Now, if we restrict the weights to non-negative values, we consider the weight function which puts one on each edge. A simple check shows that no $(2k-3)$-trail is contained in $G$.
Thus, for nonnegative weights, it remains open whether one can efficiently
compute a $(2k-1)$-trail whose weight is no more than the weight of a
cheapest solution of $\lpa$ for degree $k$.

\paragraph{\textbf{Acknowledgements.}}

We are grateful to Michel Goemans, Anupam Gupta,
Neil Olver, and Andr\'as Seb\H o
for inspiring discussions.
This research project started while both
authors were guests at the
Hausdorff Research Institute for Mathematics (HIM),
during the 2015 trimester on Combinatorial
Optimization. Both authors are very thankful to the
generous support and inspiring environment
provided by the HIM and the organizers of
the trimester program.

\appendix

\end{document}